\documentclass[a4paper,12pt]{article}
\pdfoutput=1 
\usepackage{geometry}
\usepackage[margin=10mm]{caption}

\usepackage[utf8]{inputenc}
\usepackage[english]{babel}
\usepackage{amssymb, amsmath}
\usepackage{amsthm}
\usepackage{amsfonts,bbold}
\usepackage{upgreek}
\usepackage{mathrsfs}
\usepackage{setspace}
\usepackage{graphicx}
\usepackage{subcaption}
\usepackage{slashed}
\usepackage[sorting=none,maxbibnames=99, backend=bibtex]{biblatex}
\usepackage{array,booktabs}
\usepackage{lmodern}
\usepackage{tikz}
\usepackage{microtype}
\usepackage{hyperref}
\usepackage{authblk}
\usepackage{appendix}

\usepackage{calc}

\usepackage{accents}

\newtheorem{theorem}{Theorem}
\newtheorem{prop}{Proposition}[section]
\newtheorem{lemma}{Lemma}[section]
\newtheorem{corollary}{Corollary}[lemma]

\theoremstyle{definition}
\newtheorem{remark}{Remark}[section]

\makeatletter
\renewenvironment{proof}[1][\proofname]{%
   \par\pushQED{\qed}\normalfont%
   \topsep6\p@\@plus6\p@\relax
   \trivlist\item[\hskip\labelsep\bfseries#1\@addpunct{.}]%
   \ignorespaces
}{%
   \popQED\endtrivlist\@endpefalse
}

\newcommand{\p}{\partial}
\newcommand{\ups}{\upsilon}
\newcommand{\beq}{\begin{equation}}
\newcommand{\eeq}{\end{equation}}
\newcommand{\eeeem}{\end{multline}}
\newcommand{\bem}{\begin{multline}}
\newcommand{\bqa} {\begin{eqnarray}}
\newcommand{\eqa} {\end{eqnarray}}
\newcommand{\eps}{\varepsilon}

\newcommand{\bmul}{\begin{multline}}
\newcommand{\emul}{\end{multline}}

\DeclareMathOperator{\CT}{CT}

\DeclareMathOperator{\Id}{Id}

\DeclareMathOperator{\Ad}{Ad}
\DeclareMathOperator{\ad}{ad}
\DeclareMathOperator{\dist}{dist}

\DeclareMathOperator{\End}{End}
\DeclareMathOperator{\Aut}{Aut}

\newcommand{\CA}{{\mathcal A}}
\newcommand{\CB}{{\mathcal B}}

\newcommand{\CG}{{\mathcal G}}
\newcommand{\CH}{{\mathcal H}}
\newcommand{\CI}{{\mathcal I}}

\newcommand{\CR}{{\mathcal R}}

\renewcommand{\CT}{{\mathcal T}}
\newcommand{\CU}{{\mathcal U}}
\newcommand{\CV}{{\mathcal V}}

\newcommand{\ZZ}{{\mathbb Z}}
\newcommand{\RR}{{\mathbb R}}

\newcommand{\SA}{{\mathscr A}}
\newcommand{\SC}{{\mathscr C}}
\newcommand{\SG}{{\mathscr G}}

\newcommand{\SAl}{{{\mathscr A}_\ell}}
\newcommand{\SAal}{{{\mathscr A}_{a\ell}}}

\newcommand{\SV}{{\mathscr V}}

\newcommand{\sfw}{{\mathsf w}}

\newcommand{\Or}{{O(r^{-\infty})}}
\newcommand{\On}{{O(n^{-\infty})}}

\newcommand{\AutAal}{{\Aut_{a\ell}(\SA)}}

\def \l {\left(}
\def \r {\right)}

\def \lal {\langle}
\def \ral {\rangle}

\title{An index for two-dimensional SPT states}

\author{Nikita Sopenko \smallskip\\ 
{\it California Institute of Technology, Pasadena, CA 91125, USA}}

\begin{document}

\maketitle

\abstract{We define an index for 2d $G$-invariant invertible states of bosonic lattice systems in the thermodynamic limit for a finite symmetry group $G$ with a unitary action. We show that this index is an invariant of SPT phase.}

\section{Introduction}

In the last few decades it became clear that many important physical properties of a many-body system with a gapped local Hamiltonian are encoded in the entanglement pattern of its ground state, and do not depend on the details of the Hamiltonian \cite{QImeetsQM}. In particular a gapped phase of such system, i.e. equivalence class under small deformations of the Hamiltonian which do not close the gap, can be determined from its ground state. The latter is based on the fact that the ground states of any two such systems in the same phase can be related\footnote{More precisely one should consider stable equivalence \cite{Kitaevlecture} as we define in the main text.} by an evolution by a local Hamiltonian \cite{hastings2005quasiadiabatic}, also known as quasi-adiabatic evolution.

% Importantly, the notion of a gapped phase makes sense only in the thermodynamic limit. If one wants to define invariants of phases and eventually classify them, one is forced to work with quantum states in the infinite volume.\footnote{Many distinctive properties of a gapped phase can be well demonstrated on a Hilbert space of a large but finite system. Though such approach is very insightful, from the theoretical standpoint it is not satisfactory since it does not even allow to give a definition of a phase.} Therefore characterization of such states is of fundamental importance.

A particularly simple class of states was introduced by A. Kitaev \cite{Kitaevlecture} and called invertible. By definition a state is in invertible phase if it can be tensored with some other state so that the resulting state is in the trivial phase, i.e. it can be completely disentangled by a local Hamiltonian evolution. In the presence of a symmetry group $G$ a $G$-equivariant version of such phases can be introduced, known as symmetry protected (SPT) phases \cite{gu2009tensor}. An invertible state is said to be in a non-trivial SPT phase, if it is $G$-invariant and can be inverted only by a Hamiltonian evolution that necessarily breaks $G$ symmetry\footnote{See the main text for a precise definition.}.

While one-dimensional SPT states are well understood by now\footnote{Based on earlier approaches using matrix product state approximation \cite{chen2011classification}, the invariants of phases were defined in a series of works \cite{ogata2019classification,ogata2020mathbb,bourne2020classification}. The completeness of the index for bosonic systems with respect to local Hamiltonian evolution was shown recently in \cite{KSY}.}, in higher dimensions the theory is far from complete despite a large number of examples of SPT states. A powerful approach based on group cohomology was proposed in \cite{chen2013symmetry}. However, soon after it was realized that there are states beyond this cohomology classification. It is believed that all such states are captured by a classification of invertible topological field theories \cite{kapustin2014symmetry,freed2016reflection}. It is desirable to prove this fact.

For 2d bosonic SPT states with a unitary action of a finite group $G$ both approaches predict classification by $H^{3}(G,U(1))$. In this paper we define an index for such systems taking values in $H^{3}(G,U(1))$ using the formalism developed\footnote{For an earlier development see \cite{kitaev2006anyons,kapustin2019thermal}} in \cite{kapustin2020hall} for $G=U(1)$. Our approach is motivated by the work \cite{else2014classifying}. The absence of non-trivial invertible one-dimensional states needed for the definition of the index was proven recently in the work \cite{KSY} to which this paper is complementary. 

This paper is organized as follows. In Section \ref{sec:Preliminaries} we define the setup, introduce the necessary concepts and derive some of their properties. In Section \ref{sec:index} we define the index and check that it gives expected results for states from \cite{chen2013symmetry}. The proof of the lemma needed for the definition is discussed in Appendix \ref{app:1ddefect}.
\\

\noindent
{\bf Note added:}
While this work was being written up, we became aware of
a similar result, obtained independently in \cite{ogata2021h}.

\section{Preliminaries} \label{sec:Preliminaries}

\subsection{Algebra of observables} 

By a $d$-dimensional lattice $\Lambda$ we mean a Delone set in $\RR^d$. The elements of $\Lambda$ are called sites. The distance between points $j$ and $k$ is denoted by $|j-k|$. An embedded $d$-dimensional submanifold of $\RR^d$ whose boundary has a finite number of connected components is called region. By a slight abuse of notation, we will identify a region and its intersection with the lattice $\Lambda$. The complement of a region $\Gamma$ is denoted by $\bar{\Gamma}$.

Let $\SA_{j}=\End(\SV_j)$ be a matrix algebra on a finite dimensional Hilbert space $\SV_j$, and let $\SA_{\Gamma} = \bigotimes_{j\in \Gamma} \SA_{j}$ for a finite subset $\Gamma \subset \Lambda$. A bosonic $d$-dimensional spin system is defined by its algebra of observables $\SA$ which is a $C^*$-algebra defined as a norm completion of the $*$-algebra of the form
\beq  \label{dirlim}
\SAl=\underset{\Gamma} \varinjlim\,  \SA_{\Gamma}
\eeq
The numbers $d_j = \dim \CV_j$ are assumed to be uniformly bounded by $d_j < D$. The elements of $\SA$ are called quasi-local observables, while the elements of $\SA_{\Gamma}$ for finite $\Gamma \subset \Lambda$ are called local observables localized on $\Gamma$. Following \cite{kapustin2020hall} we say that $\CA$ is an almost local observable $f$-localized at a point $p$, if for any $\CB \in \SA_k$ we have
\beq
\|[\CA,\CB]\| \leq 2 \|\CA\| \|\CB\| f(|p-k|)
\eeq
for some monotonically decreasing positive (MDP) function $f(r)=\Or$. Such observables form a dense $*$-sub-algebra $\SAal$ of $\SA$.

We denote the partial (normalized) trace over a region $\bar{B}$ by $\CA|_{B}$ which is defined on $\SAl$ by
\beq
\CA|_{B} := \int \prod_{k\in \bar{B}} d U_k \Ad_{\prod_{k \in \bar{B}} U_k} (\CA)
\eeq
with the integration over all on-site unitaries $U_k \in \SA_k$ with Haar measure, while for $\CA \in \SA$ is defined by $\lim_{n \to \infty} (\CA_{n})|_B$ for a Cauchy sequence $\{\CA_n\} \in \SAl$ with $\CA = \lim_{n \to \infty} \CA_n$.

Any $\CA \in \SAal$ almost local at site $j$ can be decomposed as an infinite sum of local observables $\CA = \sum_{n=1}^{\infty} \CA^{(n)}$ on disks $\Gamma_{n}(j)$ of radius $n$ with the center at $j$, where $\CA^{(n)}$ are defined by 
\beq
\sum_{n=1}^{m} \CA^{(n)} = \CA|_{\Gamma_m(j)}
\eeq
Note that $\| \CA^{(n)}\| \leq g(n)$ for some MDP function $g(r)=\Or$ that depends on $f(r)$ only (see Lemma A.1 of  \cite{kapustin2020hall}).

We say that a symmetry group $G$ acts on $\SA$ if for any site $j$ we have homomorphisms $\CR_j : G \to U(\CV_j)$ to a unitary group of $\CV_j$. We denote the image of $g \in G$ by $\CR^{(g)}_j$. For any (possibly infinite) region $A$ we define automorphisms $w^{(g)}_A \in \Aut(\SA)$ as a conjugation $\Ad_{\CR^{(g)}_A}$ with a (possibly formal) tensor product $\CR^{(g)}_A := \bigotimes_{j \in A} \CR^{(g)}_j$. When $A = \RR^d$ we omit $A$ and simply write $w^{(g)}$. In this paper we always assume that $G$ is a finite group.

For a state $\psi$ on $\SA$ we denote its evaluation on observables by $\lal \CA \ral_{\psi} := \psi(\CA)$. We define the action $\beta(\psi)$ of automorphisms $\beta \in \Aut(\SA)$ on states via $\lal \CA \ral_{\beta(\psi)} = \lal \beta(\CA) \ral_{\psi}$. For a region $A$ we denote the restriction of a state $\psi$ on $A$ by $\psi|_A$. We call a pure state $\psi_0$ factorized, if it satisfies $\lal \CA \CB \ral_{\psi_0} = \lal \CA \ral_{\psi_0} \lal \CB \ral_{\psi_0}$ for any observables $\CA$ and $\CB$ localized on two different sites. Sometimes, if we want to specify the algebra $\SA$ on which a given state $\psi$ is defined, we write $(\SA,\psi)$.

We say that an automorphism $\beta \in \Aut(\SA) $ is almost local, if it is also an automorphism of $\SAal$. We denote the group of such automorphisms by $\AutAal$.

\subsection{Chains}

The Hamiltonians of infinite systems are not elements of the algebra of observables. We can only define them by a formal sum
\beq
F = \sum_{j \in \Lambda} F_j
\eeq
of observables $F_j$ somehow localized at site $j$. We formalize this vague statement by requiring that $F_j$ is a self-adjoint almost local observable $f$-localized at $j$ with uniformly bounded $\|F_j\|$ (see Remark \ref{rmk:decay}). Following the terminology of \cite{kapustin2020hall} we call such objects 0-chains. We denote the group of such 0-chains under on-site addition by $\SC^0$. On $\SAal$ it defines a derivation $\ad_F(\CA) = \sum_{j \in \Lambda} [F_j,\CA] \in \SAal$. 

For a region $A$, by $F_A$ we denote a 0-chain $(F_A)_j := (\delta_{j \in A}) F_j$, where $\delta_{j \in A}$ is 1 when $j \in A$ and 0 otherwise. By $F|_{A}$ we denote a 0-chain $(F|_A)_j: = (F_j)|_{A}$.

We say that $F$ is approximately localized on a submanifold $A$, if $\|F_j\| \leq g(\dist(j,A))$ for some MDP function $g(r)=\Or$. We denote the group of such chains under on-site addition by $\SC^0_A$.

\subsection{Locally generated paths of automorphisms}

Similarly, a unitary operator that implements an evolution by a Hamiltonian does not exist as an element of the algebra of observable for infinite systems. Instead, we can only define automorphisms of the algebra of observables generated by some Hamiltonian.

Let $F(s)$ be a 0-chain, that depends on $s\in[0,1]$. We can define an automorphism $\alpha_F(s) \in \Aut(\SA)$ by
\beq
- i \frac{d}{ds} \alpha_F(s)(\CA) =\alpha_F(s)( \ad_F(\CA)), \,\,\,\,\,\,\,\,\,\,\,\,\,\alpha_F(0)=\Id
\eeq
for $\CA \in \SA_{al}$ and by extending it to the whole algebra $\SA$ (see e.g. \cite{bratteli2012operator2}). Note that $\alpha_F(s) \in \AutAal$ (see Lemma A.2 of \cite{kapustin2020hall}). We call a pair $(\alpha_F(s), F(s))$ {\it locally generated path} (LGP), and if an automorphism $\beta \in \Aut(\SA)$ coincides with $\alpha_F(1)$ for some 0-chain $F(s)$, we say that it is {\it locally generated automorphism} (LGA).

In the following by a slight abuse of notation we use $\alpha_F$ both for the pair $(\alpha_F(s),F(s))$ and for the automorphism $\alpha_{F}(1)$. We also omit $s$ in $F(s)$ and simply write $F$.

\begin{remark} \label{rmk:decay}
One may wonder why we define LGP and LGA with this particular decay for the generating Hamiltonian. The reason is twofold. First, for exact quasi-adiabatic evolution between ground states of gapped local Hamiltonians the superpolynomial decay is needed \cite{osborne2007simulating,hastings2010quasi}. Second, we want the decay to be independent of the dimension of the system, so that the classification of defect states (defined below) depends only on the dimension of the defect, not on the underlying space.
\end{remark}

LGPs form a group. Indeed, the associative composition\footnote{We use symbol $\circ$ both for composition of LGPs as defined in the main text and for composition of automorphisms of $\SA$.} $\alpha_F \circ \alpha_G$ of two LGPs $\alpha_{F}$ and $\alpha_{G}$ is an LGP generated by $G(s) + (\alpha_G(s))^{-1}(F(s))$, while the inverse $(\alpha_F)^{-1}$ of $\alpha_F$ is generated by $-\alpha_{F}(s)(F(s))$. The unit element is given by $\alpha_0$ which is a canonical LGP for the identity automorphism $\Id$. We denote this group by  $\SG(\SA)$ or simply $\SG$ if the choice of the algebra is clear from the context.

We say that $\alpha_F$ is approximately localized on a submanifold $A$, if $F \in \SC^0_A$. If $\alpha_F$ and $\alpha_G$ are approximately localized on $A$, then their composition is also approximately localized on $A$. The same is true for the inverse $(\alpha_F)^{-1}$. That allows us to define a group $\SG_A$ of LGPs approximately localized on $A$, which is a subgroup of $\SG$.

If $\alpha_F \in \SG_{p}$ for a point $p$, then $P(s) = \sum_j F_j(s)$ is an almost local observable. In this case there is a canonical unitary observable $\CV$ defined by $\CV = \CT_s \{e^{i \int_0^1 ds P(s)}\}$, where $\CT_s \{...\}$ is an ordered exponential with respect to $s$, such that $\alpha_F(1) = \Ad_{\CV}$. We say that $\CV$ is the corresponding unitary observable for $\alpha_F$. Note that if $\CV_{1}$ and $\CV_{2}$ are the corresponding unitary observables for $\alpha_{F_1},\alpha_{F_2} \in \SG_{p}$, then $(\CV_1 \CV_2)$ is the corresponding unitary observable for $\alpha_{F_1} \circ \alpha_{F_2}$. Also note that if $\CV$ is the corresponding observable for $\alpha_F \in \SG_p$, then $\alpha_{G}(\CV)$ is the corresponding observable for $\alpha_G \circ \alpha_F \circ (\alpha_G)^{-1} \in \SG_p$ (see Lemma \ref{lma:fagafa} below).

A few elementary lemmas below will be useful for us later.

\begin{lemma} \label{lma:fafa}
If $G \in \SC^0_A$, then $\alpha_{F+G} \circ (\alpha_{F})^{-1} \in \SG_A$.
\end{lemma}

\begin{proof}
$\alpha_{F+G} = \alpha_{G'} \circ \alpha_{F}$ for $G'(s) = \alpha_F(s) (G(s)) \in \SC^0_A$.
\end{proof}

\begin{lemma} \label{lma:fagafa}
If $G \in \SC^0_A$, then $\alpha_{F} \circ \alpha_G \circ (\alpha_{F})^{-1} \in \SG_A$.
\end{lemma}

\begin{proof}
Note that $\alpha_F \circ \alpha_G = \alpha_{F+G'}$ with $G'=G+((\alpha_G(s))^{-1}(F(s))-F(s)) \in \SC^0_A$. By Lemma \ref{lma:fafa} we have $\alpha_{F+G'} \circ (\alpha_{F})^{-1} \in \SG_{A}$.
\end{proof}

\begin{lemma} \label{lma:decomp}
For any given region $A$, any $\alpha_F \in \SG$ can be decomposed as
\beq
\alpha_F = \alpha_{F^{(0)}} \circ \alpha_{F|_A} \circ \alpha_{F|_{\bar{A}}}
\eeq
for some $\alpha_{F^{(0)}} \in \SG_{\p A}$.
\end{lemma}

\begin{proof}
Let $G = - (F-F|_A- F|_{\bar{A}})$. For any $j \in A$ we have $\| F_j - F_j|_A - F_j|_{\bar A} \| \leq \| F_j - F_j|_A\| + \|F_j|_{\bar A} \| \leq g(\dist(j,\p A))$ for some MDP function $g(r) = \Or$, and similarly for $j \in \bar{A}$. Therefore $G \in \SC^0_{\p A}$, and Lemma \ref{lma:fafa} implies the desired result.
\end{proof}

\begin{lemma} \label{lma:splitting}
If $\alpha_F \in \SG_{A_- \cup A_+}$ for $A_{+}$ and $A_{-}$ lying at non-intersecting cones with the same origin $p$, then we can split $\alpha_F = \alpha_{F_-} \circ \alpha_{F_+}$, such that $\alpha_{F_{\pm}} \in \SG_{A_{\pm}}$.
\end{lemma}

\begin{proof}
That is a direct consequence of the Lemma \ref{lma:decomp}
\end{proof}

\begin{lemma} \label{lma:intersection}
If $\alpha_{F_-} \circ \alpha_{F_+} = \alpha_0$ for $\alpha_{F_{\pm}} \in \SG_{A_{\pm}}$ with $A_{+}$ and $A_{-}$ lying at non-intersecting cones with the same origin $p$, then $\alpha_{F_{\pm}}  \in \SG_{p}$.
\end{lemma}

\begin{proof}
If $F_-+F_+=0$ for $F_{\pm} \in \SC^0_{A_{\pm}}$, then $F_{\pm} \in \SC^0_p$. Since $\alpha_{F_-} \circ \alpha_{F_+}$ is generated by $F_{+}(s) + (\alpha_{F_+}(s))^{-1}(F_-(s))$ satisfying this condition, we get the statement of the lemma.
\end{proof}

\subsection{States and phases}

There is a natural operation on states known as stacking. Given two state $(\SA_1,\psi_1)$ and $(\SA_2,\psi_2)$ we can construct a new state $(\SA_1 \otimes \SA_2, \psi_1 \otimes \psi_2)$. In the following we denote it by $\psi_1 \otimes \psi_2$ since it is clear that the resulting state is defined on the tensor product of the corresponding quasi-local algebras. 

We say that two states $\psi_1$ and $\psi_2$ on $\SA$ are LGA-equivalent, if there is $\alpha_F \in \SG$, such that $\psi_1 = \alpha_{F}(\psi_2)$. We call states which are LGA-equivalent to a factorized state {\it short-range entangled} (SRE). We say that two states $(\psi_1,\SA_1)$ and $(\psi_2,\SA_2)$ are stably LGA-equivalent, if there are factorized states $(\SA_1',\psi_1')$ and $(\SA_2',\psi_2')$ such that the algebras $\SA_1 \otimes \SA_1'$ and $\SA_2 \otimes \SA_2'$ are isomorphic (i.e. the lattices $\Lambda_1 \cup \Lambda'_1$ and $\Lambda_2 \cup \Lambda'_2$ are the same Delone sets and all $d_j$ are the same), and $\psi_1 \otimes \psi'_1$ and $\psi_2 \otimes \psi'_2$ are LGA-equivalent. It defines an equivalence class on states called {\it phase}. We call states which are stably LGA-equivalent to a factorized state {\it stably short-range entangled} (stably SRE).

Stacking induces a commutative monoid structure on phases with an identity being the trivial phase $\uptau$ (the phase of a factorized state). We denote this monoid by $(\Phi,\bullet,\uptau)$. Invertible elements of $\Phi$ form an abelian group $\Phi^*$ and are called {\it invertible phases}. States $\psi$ which are representatives of such phases $[\psi] \in \Phi^*$ are called {\it invertible}.

In the presence of a symmetry group $G$ we can consider a class of $G$-invariant states and define the same notions using $G$-equivariant LGA equivalence and $G$-invariant factorized state. The latter is defined to be a factorized state $\psi_0$ with $G$-invariant vectors $|v_j\ral \in \CV_j$, such that $\lal \CA_j \ral_{\psi} = \lal v_j|\CA_j|v_j \ral$ for any $\CA_j \in \SA_j$ (see \cite{KSY} for more details). Similarly, we have a monoid $(\Phi_G, \bullet, \uptau_{G})$ of $G$-invariant phases with abelian group $\Phi^*_G$ of $G$-invertible phases. We call a $G$-invariant phase {\it symmetry protected} (SPT) if is mapped to $\Phi^*$ under a forgetful map $\Phi_G \to \Phi$. Representatives of SPT phases are called {\it SPT states}.

\subsection{Defect states}

We call a state $\psi$ {\it SRE defect state} on a submanifold $A$ for a pure factorized state $\psi_0$, if it is LGA-equivalent to $\psi_0$ via $\alpha_F \in \SG_A$. 

We call a state $\psi$ {\it invertible defect state} on a submanifold $A$, if there is another system $(\SA',\psi_0')$ with a pure factorized state $\psi'_0$ and a defect state $\psi'$ on $A$ for $\psi'_0$, such that $\psi \otimes \psi'$ is an SRE defect state for $\psi_0 \otimes \psi'_0$. 

\begin{remark}
If a lattice $\Lambda$ in $\RR^n$ is a sublattice of some lattice $\tilde{\Lambda}$ in $\RR^m$ for $n<m$, then SRE and invertible states on $\Lambda$ naturally give SRE and invertible defect states on $\RR^n \subset \RR^m$.
\end{remark}

A natural way to produce an invertible defect state on a $(d-1)$-dimensional boundary $\p A$ of a region $A$ for a factorized state $\psi_0$ is to take some $\alpha_F \in \SG$, such that $\alpha_F(\psi_0) = \psi_0$, and consider $\psi = \alpha_{F_A}(\psi_0)$.

\begin{lemma} \label{lma:aFAdefect}
All such states are invertible defect states.
\end{lemma}

\begin{proof}
First, note that $(F_A - F|_A) \in \SC^0_{\p A}$, and therefore by Lemma \ref{lma:fafa} $\psi=\alpha_{F_A}(\psi_0)$ and $\alpha_{F|_A}(\psi_0)$ are LGA equivalent by an element of $\SG_{\p A}$. Second, by Lemma \ref{lma:decomp} the state $\alpha_{F|_A} \circ \alpha_{F|_{\bar{A}}} (\psi_0)$ is LGA equivalent to $\psi_0$ by an element of $\SG_{\p A}$. Let us take a copy $(\SA',\psi'_0)$ of the system $(\SA,\psi_0)$ with a defect state $\alpha_{F|_{\bar{A}}}(\psi'_0)$. Since the subalgebra $\SA_A \otimes \SA'_{\bar{A}}$ of $\SA \otimes \SA'$ is isomorphic to $\SA$, and since the states $\psi_0|_A \otimes \psi'_0|_{\bar{A}}$ and $\alpha_{F|_A}(\psi_0) \otimes \alpha_{F|_{\bar{A}}}(\psi'_0)$ under this isomorphism coincide with $\psi_0$ and $\alpha_{F|_A} \circ \alpha_{F|_{\bar{A}}} (\psi_0)$, respectively, the state $\alpha_{F|_A}(\psi_0) \otimes \alpha_{F|_{\bar{A}}}(\psi'_0)$ is LGA equivalent to $\psi_0 \otimes \psi'_0$ by an element of $\SG_{\p A}(\SA \otimes \SA')$.
\end{proof}
\noindent

\begin{remark}
In fact in this way any invertible state for a lattice $\Lambda \subset \RR^{n}$ can be realized as an invertible defect state for some lattice in $\RR^{n+1}$ on the boundary of a half-plane. For that we can take a lattice that coincides with $\Lambda$ at even values of $x_{n+1}$ and with $\Lambda'$ (the lattice of the inverse system) at odd values of $x_{n+1}$. The automorphism $\alpha_F$ can be constructed using Eilenberg swindle.
\end{remark}

From now on we consider two dimensional lattices $\Lambda \subset \RR^2$ only. We define a cone-like region $A$ with the origin at the point $p$ to be a region which asymptotically behaves as a cone with positive angle (see Fig. \ref{fig:conedef}).

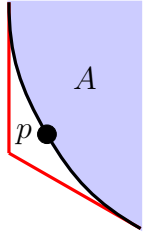
\begin{figure}
\centering
\begin{tikzpicture}[scale=.5]

\draw [color=blue!20, fill=blue!20, very thick] (1,0.5) to [out=120,in=-90] (0,4) -- (3.4641,4) -- (3.4641,-2) to [out=150,in=-60] (1,0.5);

\draw[red, very thick] (0,0) -- (3.4641,-2);
\draw[red, very thick] (0,0) -- (0,4);

\filldraw[black] (1,0.5) circle (7pt) node[anchor=west]{};

\node  at (.4,.5) {$p$}; 
\node  at (2,2) {$A$};

\coordinate (G) at (1,0.5);
\coordinate (R) at (0,4);
\coordinate (B) at (3.4641,-2);
\draw[black, very thick] (G) to [out=120,in=-90] (R);
\draw[black, very thick] (G) to [out=-60,in=150] (B);

\end{tikzpicture}
\caption{Definition of a cone-like region. 
}
\label{fig:conedef}
\end{figure}

We say that two states  $\psi_1$ and $\psi_2$ on $C$ are $f$-close far from $j$ for some MDP function $f(r) = \Or$, if for any observable $\CA \in \SA_{A \cup \bar{\Gamma}_r(j)}$ with $\Gamma_r(j)$ being a disk of radius $r$ with the center at $j$ we have $\|\psi_1(\CA) - \psi_2 (\CA)\| \leq \|\CA\| f(r)$.

Using the results of \cite{KSY} we have the following lemma, which we prove in the appendix.

\begin{lemma}\label{lma:1ddefect}
Let $\psi_0$ be a 2d pure factorized state, and let $\alpha_Q \in \SG$, such that $\alpha_Q(\psi_0) = \psi_0$. Then for any cone-like region $A$ the state $\alpha_{Q_{A}}^{-1}(\psi_0)$ is SRE defect state.
\end{lemma}

\begin{corollary} \label{cor:0ddefect}
All invertible defect states at a point are SRE defects states.
\end{corollary}
\noindent
The last corollary is also a direct consequence of Lemma \ref{lma:almfac}.

\begin{corollary} \label{cor:localSSB}
Let $\psi$ be a 2d SRE state, and let $\alpha_Q \in \SG$, such that $\alpha_Q(\psi) = \psi$. Then for any cone-like region $A$ there is $K \in \SG_{\p A}$, such that $(\alpha_{Q_{A}} \circ \alpha_{K})(\psi) = \psi$. 
\end{corollary}

\begin{proof}
Note that $(\alpha_F \circ \alpha_{Q_A} \circ (\alpha_F)^{-1})(\psi) = \alpha_{Q'_A}(\psi)$ for $Q' = \alpha_F(Q)$ and any state $\psi$. If $\psi = \alpha_F(\psi_0)$ for some factorized state $\psi_0$, then $(\alpha_F \circ \alpha_{Q} \circ (\alpha_F)^{-1} )(\psi_0) = \alpha_{Q'}(\psi_0) = \psi_0$.  By Lemma \ref{lma:aFAdefect} and Lemma \ref{lma:1ddefect} there is $\alpha_{K'} \in \SG_{\p A}$ such that $(\alpha_{Q'_A} \circ \alpha_{K'})(\psi_0) = \psi_0$. Therefore we can take $\alpha_K = (\alpha_F)^{-1} \circ \alpha_{K'} \circ \alpha_{F}$.
\end{proof}

\begin{corollary} \label{cor:localSSB0d}
Let $\psi$ be a 2d SRE state, and let $\alpha_Q \in \SG_{\p A}$ for a cone-like region $A$ with the origin at $p$, such that $\alpha_Q(\psi) = \psi$. Then for a splitting $Q=Q_- + Q_+$, such that $Q_{\pm} \in \SC^0_{(\p A)_{\pm}}$, there is $N \in \SC^0_{p}$, such that $(\alpha_{Q_-} \circ \alpha_N)(\psi) = \psi$. 
\end{corollary}

\begin{proof}
Let $B$ be a cone containing $(\p A)_+$ and not containing $(\p A)_-$. As in the proof of Corollary \ref{cor:localSSB} let $Q' = \alpha_F(Q) \in \SC^0_{\p A}$ which satisfies $\alpha_{Q'}(\psi_0) = \psi_0$, and let $Q'_{\pm} = \alpha_F(Q_{\pm}) \in \SC^0_{(\p A)_{\pm}}$. We have $\alpha_{M} := \alpha_{Q'} \circ \alpha_{Q'_{+}}^{-1} \circ \alpha_{Q'_{-}}^{-1} \in \SG_{p}$ and $\alpha_{Q'_{-}}^{-1}(\psi_0) = (\alpha_{M} \circ \alpha_{Q'_{+}})(\psi_0)$. It implies that $\alpha_{Q'_{-}}^{-1}(\psi_0)$ is $g$-close to $\psi_0$ for some $g(r)=\Or$ both on $B$ and on $\bar{B}$, that by Lemma \ref{lma:almfacgluing} implies that it is $h$-close to $\psi_0$ on $\RR^2$ for some $h(r)=\Or$. By Lemma \ref{lma:almfac} there is $\alpha_{N'}$, such that $(\alpha_{Q'_-} \circ \alpha_{N'}) (\psi_0) = \psi_0$. Then $\alpha_N = (\alpha_F)^{-1} \circ \alpha_{N'} \circ \alpha_{F}$.
\end{proof}

\section{The index for 2d SPT states} \label{sec:index}

\subsection{Definition} \label{ssec:IndexDef}

\begin{figure}
\centering
\begin{tikzpicture}[scale=.5]

\draw [color=blue!20, fill=blue!20, very thick] (0,0) to (0,4) -- (3.4641,4) -- (3.4641,-2) -- (0,0);

\draw[black, very thick] (0,0) -- (3.4641,-2);
\draw[black, very thick] (0,0) -- (0,4);

\filldraw[black] (0,0) circle (7pt) node[anchor=west]{};

\node  at (-.5,-.5) {$p$}; 
\node  at (2,1.5) {$A$};
\node  at (0,4.5) {$(\p A)_-$};
\node  at (1.3*3.4641,-1.3*2) {$(\p A)_+$};

\end{tikzpicture}
\caption{$A$ is a cone-like region at $p$ with boundary components $(\p A)_-$ and $(\p A)_{+}$. 
}
\label{fig:coneA}
\end{figure}
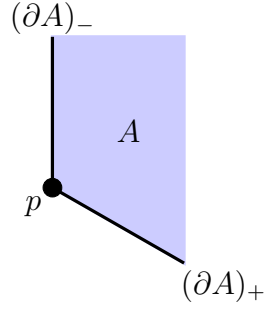

Let $\psi$ be a 2d SPT state on $\SA$ with a symmetry group $G$. By definition we can stack it with a state $(\SA',\psi')$ with a trivial $G$-action, so that the resulting state is SRE. We redefine $\psi$ by $\psi \otimes \psi'$ for the rest of this subsection. Let $\psi_0$ be a factorized pure state such that  $\psi=\alpha_F(\psi_0)$ for some $\alpha_F \in \SG$. Let $w^{(g)}$ be an automorphism, that corresponds to the on-site symmetry action. Their restrictions to any region $A$ satisfy the evident relations
\beq \label{eq:evrel1}
w_A^{(g)} \circ w_A^{(h)} \circ \l w_A^{(g h)} \r^{-1} = \Id,
\eeq
\beq \label{eq:evrel2}
w_A^{(g)} \circ w_A^{(h)} \circ w_A^{(k)} \circ \l w_A^{(g h k)} \r^{-1} = \Id,
\eeq
\beq \label{eq:evrel3}
w_A^{(g)} \circ w_A^{(h)} \circ w_A^{(k)} \circ w_A^{(l)} \circ \l w_A^{(g h k l)} \r^{-1} = \Id .
\eeq
Note that the automorphism $w^{(g)}$ is an LGA, but there is no canonical LGP for it.

We introduce automorphisms $\sfw^{(g)}$ of the group $\SG$ defined by $\sfw^{(g)}(\alpha_F) : = \alpha_{w^{(g)}(F)}$. Similarly, we define $\sfw^{(g)}_A(\alpha_F) : = \alpha_{w^{(g)}_A(F)}$.

Let $A$ be a cone-like region as shown on fig. \ref{fig:coneA}. Let us split $w^{(g)}$ into $w^{(g)}_- \circ w^{(g)}_+$ with $w^{(g)}_+ = w^{(g)}_A$ and $w^{(g)}_- = w^{(g)}_{\bar{A}}$. Since $w^{(g)}(\psi) = \psi$, and since $w^{(g)}$ is can be locally generated by an on-site 0-chain, by Corollary \ref{cor:localSSB} there is $\alpha_{K^{(g)}}^{-1} \in \SG_{\p A}$, such that 
\beq 
\l w_{-}^{(g)} \circ \alpha_{K^{(g)}}^{-1}\r(\psi) = \psi = \l \alpha_{K^{(g)}} \circ  w_{+}^{(g)} \r(\psi).
\eeq 
We define an automorphism $\tilde{\sfw}_{+}^{(g)}$ of $\SG$ by
\beq
\tilde{\sfw}_{+}^{(g)}(\alpha_{F}) := \alpha_{K^{(g)}} \circ \sfw_{+}^{(g)}(\alpha_F) \circ (\alpha_{K^{(g)}})^{-1}
\eeq
which by Lemma \ref{lma:fagafa} is also an automorphism of $\SG_{B}$ for any submanifold $B$. We define an element of $\SG_{\p A}$ by
\beq
\ups^{(g,h)} :=  \alpha_{K^{(g)}} \circ \sfw_+^{(g)}(\alpha_{K^{(h)}}) \circ (\alpha_{K^{(gh)}})^{-1}
\eeq
which depends on the choice of $K$ only. Since $w_+^{(g)} \circ w_+^{(h)} \circ (w_+^{(gh)})^{-1} = \Id$ we have $\ups^{(g,h)}(\psi)=\psi$. By Lemma \ref{lma:splitting} it can be split as
\beq
\ups^{(g,h)} = \ups_{-}^{(g,h)} \circ \ups_{+}^{(g,h)}
\eeq 
for $\ups_{\pm}^{(g,h)} \in \SG_{(\p A)_{\pm}}$. 

Since $\ups^{(g,h)}(\psi)=\psi$, by Corollary \ref{cor:localSSB0d} there is $\alpha_{N^{(g,h)}}^{-1} \in \SG_p$, such that 
\beq 
\l \ups_{-}^{(g,h)} \circ \alpha_{N^{(g,h)}}^{-1} \r = \psi = \l \alpha_{N^{(g,h)}} \circ \ups_{+}^{(g,h)} \r (\psi).
\eeq
We define LGPs
\beq
\tilde{\ups}_{+}^{(g,h)} := \alpha_{N^{(g,h)}} \circ \ups_{+}^{(g,h)},
\eeq
\beq
\tilde{\ups}_{-}^{(g,h)} := \ups_{-}^{(g,h)} \circ (\alpha_{N^{(g,h)}})^{-1}
\eeq
which preserve $\psi$ and give another decomposition 
\beq
\ups^{(g,h)} = \tilde{\ups}^{(g,h)}_- \circ \tilde{\ups}^{(g,h)}_+.
\eeq
Using an identity on LGPs
\beq
\ups^{(g,h)} \circ \ups^{(gh,k)} \circ (\ups^{(g,hk)} )^{-1} \circ \l \tilde{\sfw}_+^{(g)} (\ups^{(h,k)}) \r^{-1} = \alpha_{0}
\eeq
and Lemma \ref{lma:intersection} we can define an LGP 
\beq
\iota^{(g,h,k)} = \tilde{\ups}_{+}^{(g,h)} \circ \tilde{\ups}_{+}^{(gh,k)} \circ \l \tilde{\ups}_{+}^{(g,hk)} \r^{-1} \circ \l \tilde{\sfw}_+^{(g)} (\tilde{\ups}_{+}^{(h,k)})\r^{-1}
\eeq
which is an element of $\SG_p$ and depends on the choice of $K$ and $N$ only. Note that $\iota^{(g,h,k)}(\psi) = \psi$. Since $\iota^{(g,h,k)} \in \SG_p$, it has the corresponding unitary observable. We denote it by $\CI^{(g,h,k)}$. 

For the GNS representation $(\Pi_{\psi},\CH_{\psi},|\psi\ral)$ of $\psi$ the condition $\Ad_{\CU}(\psi) = \psi$ for a unitary observable $\CU$ (in particular $\CI^{(g,h,k)}$) implies that $\Pi_{\psi}(\CU)^* | \psi \ral$ coincides with $|\psi\ral$ up to a phase. Therefore, $|\lal \CU \ral_{\psi}|=1$ and $\lal \CA \CU \ral_{\psi} = \lal \CA \ral_{\psi} \lal \CU \ral_{\psi}$ for any observable $\CA \in \SA$.

We define
\beq
\omega(g,h,k) := \lal \CI^{(g,h,k)} \ral_{\psi}.
\eeq
that takes values in $U(1)$.

We have the following relations between LGPs
\begin{multline}  \label{eq:AUTpentagon}
\iota^{(g,h,k)} \circ (\tilde{\sfw}_+^{(g)}(\tilde{\ups}_+^{(h,k)})) \circ \iota^{(g,hk,l)} \circ (\tilde{\sfw}_+^{(g)}(\tilde{\ups}_+^{(hk,l)})) \circ \tilde{\ups}_+^{(g,hkl)} = \\
= \iota^{(g,h,k)} \circ (\tilde{\sfw}_+^{(g)}(\tilde{\ups}_+^{(h,k)})) \circ \tilde{\ups}_+^{(g,hk)} \circ \tilde{\ups}_+^{(ghk,l)} = \\
\tilde{\ups}_+^{(g,h)} \circ \tilde{\ups}_+^{(gh,k)} \circ \tilde{\ups}_+^{(ghk,l)} = \\
= \tilde{\ups}_+^{(g,h)} \circ \iota^{(gh,k,l)} \circ (\tilde{\sfw}_+^{(gh)} (\tilde{\ups}_+^{(k,l)} )) \circ \tilde{\ups}_+^{(gh,kl)} = \\
= \tilde{\ups}_+^{(g,h)} \circ \iota^{(gh,k,l)} \circ (\tilde{\sfw}_+^{(gh)} (\tilde{\ups}_+^{(k,l)}) ) \circ (\tilde{\ups}_+^{(g,h)})^{-1} \circ \iota^{(g,h,kl)} \circ (\tilde{\sfw}_+^{(g)} (\tilde{\ups}_+^{(h,kl)})) \circ \tilde{\ups}_+^{(g,hkl)}.
\end{multline}
Multiplying the last line by the inverse of the first and using
\beq
\tilde{\ups}_{+}^{(hk,l)} \circ \l \tilde{\ups}_{+}^{(h,kl)} \r^{-1} = (\tilde{\ups}_{+}^{(h,k)})^{-1} \circ (\iota^{(h,k,l)})^{-1} \circ \l \tilde{\sfw}_+^{(h)} (\tilde{\ups}_{+}^{(k,l)}) \r^{-1}
\eeq
we get the pentagon relation on LGPs
\begin{multline}
\tilde{\ups}_+^{(g,h)} \circ \iota^{(gh,k,l)} \circ (\tilde{\sfw}_+^{(gh)} (\tilde{\ups}_+^{(k,l)}) ) \circ (\tilde{\ups}_+^{(g,h)})^{-1} \circ \\ \circ \iota^{(g,h,kl)} \circ \tilde{\sfw}_+^{(g)} \l ( \tilde{\ups}_+^{(h,k)})^{-1} \circ (\iota^{(h,k,l)})^{-1} \circ \l \tilde{\sfw}_+^{(h)} (\tilde{\ups}_{+}^{(k,l)})\r^{-1} \r  \\  \circ (\iota^{(g,hk,l)})^{-1} \circ \l \tilde{\sfw}_+^{(g)}(\tilde{\ups}_+^{(h,k)})\r^{-1} \circ (\iota^{(g,h,k)})^{-1} = \alpha_0
\end{multline}
In the GNS representation $(\Pi_{\psi},\CH_{\psi},|\psi\ral)$ of $\psi$ the LGAs corresponding to $\alpha_0, \iota^{(g,h,k)}$, $\tilde{\ups}_{+}^{(g,h)}$, $\alpha_{K^{(g)}} \circ w_{+}^{(g)}$ are represented by conjugations with unitary operators on $\CH_{\psi}$ preserving $|\psi\ral \lal \psi|$, and therefore acting on $|\psi\ral$ by multiplication by a phase. We can take this unitaries for $\iota^{(g,h,k)}$ and $\alpha_0$ to be $\Pi_{\psi}(\CI^{(g,h,k)})$ and $\Pi_{\psi}(1)$, respectively. The pentagon relation on LGPs then gives a pentagon relation on the corresponding unitaries, and since the phase ambiguity for unitaries corresponding to $\tilde{\ups}_{+}^{(g,h)}$, $\alpha_{K^{(g)}} \circ w_{+}^{(g)}$ cancels in this relation, by taking the expectation value $\lal \psi| \cdot |\psi \ral$ we get
% Since for unitary observables $\CU_a$, such that $\Ad_{\CU_a}(\psi) = \psi$, we have
% \beq \label{eq:Urel1}
% \lal \CU_a \ral_{\psi} = \lal \tilde{\ups}_{+}^{(g,h)}(\CU_a) \ral_{\psi} = \lal \alpha_{K^{(g)}} \circ w_+^{(g)}(\CU_a) \ral_{\psi}
% \eeq
% \beq \label{eq:Urel2}
% \lal \CU_1 \CU_2 ... \CU_{n} \ral_{\psi} = \lal \CU_1 \ral_{\psi} \lal \CU_2 \ral_{\psi} ... \lal \CU_n \ral_{\psi}
% \eeq
% from the pentagon relation we get
\beq
\omega(g,h,k) \omega(g,hk,l) \omega(h,k,l) = \omega(gh,k,l) \omega(g,h,kl)
\eeq
For unitary observables $\CU_a$, such that $\Ad_{\CU_a}(\psi) = \psi$, we also have
\beq \label{eq:Urel1}
\lal \CU_a \ral_{\psi} = \lal \tilde{\ups}_{+}^{(g,h)}(\CU_a) \ral_{\psi} = \lal \alpha_{K^{(g)}} \circ w_+^{(g)}(\CU_a) \ral_{\psi}
\eeq
\beq \label{eq:Urel2}
\lal \CU_1 \CU_2 ... \CU_{n} \ral_{\psi} = \lal \CU_1 \ral_{\psi} \lal \CU_2 \ral_{\psi} ... \lal \CU_n \ral_{\psi}
\eeq

\begin{theorem}
The cohomology class $[\omega]\in H^3(G,U(1))$ depends on $\psi$ only, i.e. it does not depend on the choice of $K^{(g)}$ and $N^{(g,h)}$ and the cone-like region $A$.
\end{theorem}

\begin{proof}
Any change of $N^{(g,h)}$ corresponds to a redefinition $\tilde{\ups}^{(g,h)}_{+} \to \alpha_{N'^{(g,h)}} \circ \tilde{\ups}^{(g,h)}_{+}$ for some $\alpha_{N'^{(g,h)}} \in \SG_p$, such that $\alpha_{N'^{(g,h)}}(\psi) = \psi$. Let $\CV^{(g,h)}$ be the corresponding unitary observable for $\alpha_{N'^{(g,h)}}$. Using eq. (\ref{eq:Urel1}) and eq. (\ref{eq:Urel2}) we get
\beq
\omega(g,h,k) \to \omega(g,h,k) \frac{\mu(h,k) \mu(g,hk)}{\mu(gh,k) \mu(g,h)}
\eeq
for $\mu(g,h) = \lal \CV^{(g,h)} \ral_{\psi} $, so that $[\omega]$ is not affected.

Since $[\omega]$ in independent of $N^{(g,h)}$, we can change the position of $p$ and the splitting $\ups^{(g,h)} = \ups^{(g,h)}_{-} \circ \ups^{(g,h)}_{+}$ to any point on $\p A$ without changing $[\omega]$. 

Let $C_+$ be a cone containing $(\p A)_{+}$. Suppose we change $K^{(g)}$ or the boundary $(\p A)_{+}$ itself (so that we still have a cone-like region $A$) inside $C_+$. We can take $p$ to be at distance $r$ from $C_+$, and such change can only affect $\omega(g,h,k)$ by $\Or$ terms. Since $r$ can be arbitrary large, this term is actually vanishing. Similarly, any change of $K^{(g)}$ or $(\p A)_-$ inside the complement of $C_+$ can't change $[\omega]$.
\end{proof}

\begin{remark}
Note that our index is defined only using the algebra of observables and structural properties of invertible states of lower dimensions. We don't have to introduce any representation of this algebra.
\end{remark}

\begin{remark}
The structural properties of $d$-dimensional invertible states for $d>1$ are not fully understood. That gives an obstruction for us to define this index for $d$-dimensional systems with $d>2$, since it is believed that in $d>1$ there are invertible states which are not SRE (e.g. a conjectural Kitaev's $E_8$ state \cite{kitaev2006anyons}). However, this is not surprising since it is known that there are states which are not captured by cohomology classification in higher dimensions even for unitary symmetries (see \cite{fidkowski2020exactly}).
\end{remark}

\subsection{Basic properties}

\begin{prop}
The index is multiplicative under stacking, i.e. the index of $(\SA_{12},\psi_{12}) = (\SA_1,\psi_1) \otimes (\SA_2, \psi_2)$ is a product of indices for $(\SA_1,\psi_1)$ and $(\SA_2,\psi_2)$.
\end{prop}
\begin{proof}
For a stack $(\SA_{12},\psi_{12}) = (\SA_1,\psi_1) \otimes (\SA_2, \psi_2)$ of two different states $(\SA_1,\psi_1)$ and $(\SA_2, \psi_2)$ we have
\beq
\lal \CI_{12}^{(g,h,k)} \ral_{\psi_{12}} = \lal (\CI_1^{(g,h,k)} \otimes \CI_2^{(g,h,k)}) \ral_{\psi_{12}} = \lal \CI_1^{(g,h,k)} \ral_{\psi_1} \lal \CI_2^{(g,h,k)} \ral_{\psi_2},
\eeq
\end{proof}

\begin{prop}
Any two states in the same SPT phase have the same index.
\end{prop}
\begin{proof}
First, note that the index does not depend on the state $\psi'$ with a trivial $G$-action chosen at the beginning of subsection \ref{ssec:IndexDef} to produce SRE state $\psi \otimes \psi'$. Indeed, if there is another such state $\psi''$, we can realize both computations of the index using $\psi \otimes \psi'$ and $\psi \otimes \psi''$ on a system $\psi \otimes \psi' \otimes \psi'' \otimes \psi'''$, where $\psi'''$ is a copy of $\psi$ with a trivial $G$-action.

Second, the index is not affected if we stack the state $\psi$ with a $G$-invariant pure factorized state, since the index is multiplicative and the latter has the trivial index. It is left to show that the index is not affected by $G$-equivariant LGP.

Let $\alpha_F$ be a $G$-equivariant LGP, such that $\psi = \alpha_F(\psi')$ for $G$-invariant SPT states $\psi$ and $\psi'$. Note that $\alpha_{F} \circ \sfw^{(g)}_{+}(\alpha_F)^{-1} \in \SG_{\p A}$. Suppose we have a choice of $K^{(g)}$ and $N^{(g,h)}$ for the state $\psi$. For the state $\psi'$ we can choose 
\beq
\alpha_{K'^{(g)}} = \alpha_F \circ \alpha_{K^{(g)}} \circ \sfw_{+}^{(g)}(\alpha_{F})^{-1} \in \SG_{\p A},
\eeq
\beq
\ups'^{(g,h)}_{+} = \alpha_F \circ \ups^{(g,h)}_{+} \circ (\alpha_{F})^{-1} \in \SG_{(\p A)_+},
\eeq
\beq
\alpha_{N'^{(g,h)}} = \alpha_F \circ \alpha_{N^{(g,h)}} \circ (\alpha_{F})^{-1} \in \SG_p.
\eeq
Then for the state $\psi'$ we have $\iota'^{(g,h,k)} = \alpha_F \circ \iota^{(g,h,k)} \circ (\alpha_{F})^{-1}$ with the corresponding observable $\CI'^{(g,h,k)} = \alpha_F(\CI^{(g,h,k)})$. Therefore, the index of the state $\psi'$
\beq
\lal \CI'^{(g,h,k)} \ral_{\psi'} = \lal \alpha_F(\CI^{(g,h,k)}) \ral_{\psi'} = \lal \CI^{(g,h,k)} \ral_{\psi}.
\eeq
\end{proof}

\subsection{Example}

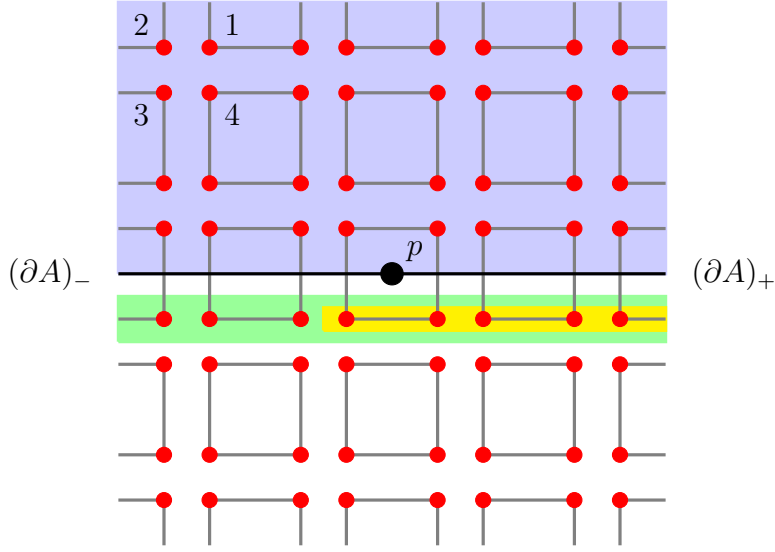
\begin{figure}
\centering
\begin{tikzpicture}[scale=.6]

\draw [color=blue!20, fill=blue!20, very thick] (-6,0) to (6,0) -- (6,6) -- (-6,6) -- (-6,0);

\draw [color=green!40, fill=green!40, very thick] (-6,-1.5) to (6,-1.5) -- (6,-0.5) -- (-6,-0.5) -- (-6,-1.5);

\draw [color=yellow!100, fill=yellow!100, very thick] (-1.5,-1.25) to (6,-1.25) -- (6,-0.75) -- (-1.5,-0.75) -- (-1.5,-1.25);

\foreach \x in {-5,-2,...,4}{                           % Two indices running over each
    \foreach \y in {-4,-1,...,2}{                       % node on the grid we have drawn 
    \draw[gray, very thick] (\x,\y) -- (\x,\y+2);
    \draw[gray, very thick] (\x+1,\y) -- (\x+1,\y+2);
    }
}

\foreach \x in {-4,-1,...,2}{                           % Two indices running over each
    \foreach \y in {-5,-2,...,5}{                       % node on the grid we have drawn 
    \draw[gray, very thick] (\x,\y) -- (\x+2,\y);
    \draw[gray, very thick] (\x,\y+1) -- (\x+2,\y+1);
    }
}

\foreach \x in {-5,-2,...,4}{                           %
    \draw[gray, very thick] (\x,5) -- (\x,6);
    \draw[gray, very thick] (\x+1, 5) -- (\x+1, 6);
    \draw[gray, very thick] (\x,-5) -- (\x,-6);
    \draw[gray, very thick] (\x+1,-5) -- (\x+1,-6);
    \draw[gray, very thick] (-5,\x) -- (-6,\x);
    \draw[gray, very thick] (-5,\x+1) -- (-6,\x+1);
    \draw[gray, very thick] (5,\x) -- (6,\x);
    \draw[gray, very thick] (5,\x+1) -- (6,\x+1);
}

\draw[black, very thick] (-6,0) -- (6,0);

\node  at (0.5,0.5) {$p$};

\node  at (7.5,0) {$(\p A)_+$};
\node  at (-7.5,0) {$(\p A)_-$};

\node  at (-5.5,5.5) {$2$}; 
\node  at (-3.5,5.5) {$1$}; 
\node  at (-3.5,3.5) {$4$}; 
\node  at (-5.5,3.5) {$3$}; 

\filldraw[black] (0,0) circle (7pt) node[anchor=west]{};

\foreach \x in {-5,-2,...,4}{                           % Two indices running over each
    \foreach \y in {-5,-2,...,4}{                       % node on the grid we have drawn 
    \node[draw,red,circle,inner sep=2pt,fill] at (\x,\y) {}; % Places a dot at those points
    \node[draw,red,circle,inner sep=2pt,fill] at (\x+1,\y) {}; % Places a dot at those points
    \node[draw,red,circle,inner sep=2pt,fill] at (\x+1,\y+1) {}; % Places a dot at those points
    \node[draw,red,circle,inner sep=2pt,fill] at (\x,\y+1) {}; % Places a dot at those points
    }
}

\end{tikzpicture}
\caption{ The automorphisms $w^{(g)}$ act non-trivially on the blue shaded region. The LGPs $\alpha_{K^{(g)}}$, $\ups^{(g,h)}$ act non-trivially on the green shaded region. The LGPs $\tilde{\ups}^{(g,h)}_+$ act non-trivially on the yellow shaded region.
}
\label{fig:TNstate}
\end{figure}

As was shown in \cite{chen2013symmetry}, for any representative of a class $[\omega] \in H^{3}(G,U(1))$ one can construct a tensor network state that is SPT. Let us show that our index for such state is $[\omega]$.

Let $\Lambda$ be a square lattice with sites $(x,y) \in \ZZ \times \ZZ$. Let $\CV$ be a regular representation of $G$ with a basis $\lal l |$ for $l \in G$. The on-site Hilbert space is $\CV_{(x,y)} = \bigotimes_{a=1}^{4} \CV^{(a)}_{(x,y)}$ for $\CV^{(a)}_{(x,y)} \cong \CV$. The induced basis for $\CV_{(x,y)}$ is 
\beq
\lal l_1, l_2, l_3, l_4 | := \lal l_1 |^{(1)} \otimes \lal l_2 |^{(2)} \otimes \lal l_3 |^{(3)} \otimes \lal l_4 |^{(4)}.
\eeq
The on-site action of $G$ is defined by
\beq
\lal l_1,l_2,l_3,l_4 | \CR^{(g)}  = \lal l_1 g, l_2 g, l_3 g, l_4 g | \frac{\omega(l_2 l_1^{-1}, l_1, g)\omega(l_3 l_2^{-1}, l_2, g)}{\omega(l_3 l_4^{-1}, l_4, g)\omega(l_4 l_1^{-1}, l_1, g)}
\eeq
for some a representative $\omega(g,h,k)$ of $[\omega]$. The state $\psi$ is chosen to be a tensor product of vector states with
\beq
\lal \psi_{(x-1/2,y-1/2)} | \sim \sum_{l \in G} \lal l|^{(1)}_{(x,y)} \otimes \lal l|^{(2)}_{(x-1,y)} \otimes \lal l|^{(3)}_{(x-1,y-1)} \otimes \lal l|^{(4)}_{(x,y-1)} 
\eeq
on $\CV^{(1)}_{(x,y)} \otimes \CV^{(2)}_{(x-1,y)} \otimes \CV^{(3)}_{(x-1,y-1)} \otimes \CV^{(4)}_{(x,y-1)}$. One can check that such vector is invariant under the on-site action of $G$ on a region that contains four involved sites. 

Let $A$ be an upper half-plane $y>0$, so that $w_{+}^{(g)}$ acts on sites $(x,y) \in \ZZ \times \ZZ_{>0}$. Let $\CV^{link}_x$ be a diagonal subspace of $\CV^{(1)}_{(x-1,0)} \otimes \CV^{(2)}_{(x,0)}$. Let $ K''^{(g)}_x$ be commuting self-adjoint observables localized on $\CV^{link}_{x} \otimes \CV^{link}_{x+1}$ such that for basis vectors
\beq
\lal l_{x}| \otimes \lal l_{x+1}| e^{i K''^{(g)}_x} = \lal l_{x}| \otimes \lal l_{x+1}| \omega(l_{x} l_{x+1}^{-1}, l_{x+1}, g)
\eeq
Let $ K'^{(g)}_x$ be self-adjoint observables localized on $\CV^{link}_{x}$ such that for basis vectors
\beq
\lal l_{x}| e^{i K'^{(g)}_x} = \lal l_{x} g|
\eeq
We can take $\alpha_{K^{(g)}} = \alpha_{\sum_x K''^{(g)}_x} \circ \alpha_{\sum_x K'^{(g)}_x}$.
Then an LGP
\beq
\ups^{(g,h)} = \alpha_{K^{(g)}} \circ \sfw^{(g)}(\alpha_{K^{(h)}}) \circ (\alpha_{K^{(gh)}})^{-1} = \alpha_{K^{(g)}} \circ \alpha_{K^{(h)}} \circ (\alpha_{K^{(gh)}})^{-1}
\eeq
gives an LGA that can be described as a conjugation with unitaries $e^{iQ'^{(g,h)}_x}$ such that for basis vectors
\beq
\lal l_{x}| \otimes \lal l_{x+1} | e^{i Q'^{(g,h)}_x}  = \lal l_{x}| \otimes \lal l_{x+1} | \frac{\omega(l_{x},g,h)}{\omega(l_{x+1},g,h)}
\eeq
where we have used
\beq
\frac{\omega(l_{x} l_{x+1}^{-1}, l_{x+1} g, h) \omega(l_{x} l_{x+1}^{-1}, l_{x+1}, g)}{\omega(l_{x} l_{x+1}^{-1}, l_{x+1}, g h)} = \frac{\omega(l_{x},g,h)}{\omega(l_{x+1},g,h)}.
\eeq
We can take 
\beq 
\ups^{(g,h)}_+ := \alpha_{K^{(g)}_+} \circ \alpha_{K^{(h)}_+} \circ (\alpha_{K^{(gh)}_+})^{-1}
\eeq
where $\alpha_{K^{(g)}_+} = \alpha_{\sum_{x\geq 0} K''^{(g)}_x} \circ \alpha_{\sum_{x\geq0} K'^{(g)}_x}$, and define the corresponding $\alpha_{N^{(g,h)}}$ local on $\CV^{link}_0$ by
\beq
\lal l_0| e^{i N^{(g,h)}}  = \lal l_0| \omega(l_0,g,h)^{-1}.
\eeq
Finally, using
\beq
\frac{\omega(l_0 g,h,k) \omega(l_0,g,hk)}{\omega(l_0,g,h) \omega(l_0,gh,k)} = \omega(g,h,k)
\eeq
we get $\iota^{(g,h,k)}$ with the corresponding unitary $\CI^{(g,h,k)} = \omega(g,h,k)$, that gives
\beq
\lal \CI^{(g,h,k)} \ral_{\psi} = \omega(g,h,k).
\eeq
\\

\noindent
{\bf Acknowledgements:}
I would like to thank Anton Kapustin for many inspiring discussions and comments on the draft. This research was supported in part by the U.S.\ Department of Energy, Office of Science, Office of High Energy Physics, under Award Number DE-SC0011632. 
\\

\noindent
{\bf Data availability statement:}
Data sharing is not applicable to this article as no new data were created or analyzed in this study.

\appendix
\numberwithin{equation}{section}

\section{Proof of Lemma \ref{lma:1ddefect}}\label{app:1ddefect}

Since the proof of Lemma \ref{lma:1ddefect} is essentially the same as the proof that all invertible states on a one-dimensional lattice are stably SRE from Section 4 of \cite{KSY}, we just sketch the differences, assuming the reader is familiar with this proof and the terminology. 

% All lemmas mentioned below are from Section 4 of \cite{KSY}.

We also assume that $A$ is a half-plane $y>0$ for simplicity of the exposition. It is clear from the proof that for a cone-like region all the arguments below can be generalized.

\begin{lemma} \label{lma:almfac}
Let $\psi$ be a pure state which is $f$-close on $\RR^2$ far from $j$ to a pure factorized state $\psi_0$ for some MDP function $f(r)=\Or$. Then $\psi$ and $\psi_0$ are unitarily equivalent and one can be produced from the other by a conjugation with $e^{i \CG}$, where $\CG$ is an almost local self-adjoint observable $g$-localized at $j$ and bounded $\|\CG\|\leq C$  for some $g(r)$ and $C$ which only depend on $f(r)$.
\end{lemma}

\begin{proof}
The proof word to word repeats the proof of Lemma 4.1 from \cite{KSY} with  $\Gamma_r(j)$ now being the disk of radius $r$ with the center at $j$.
\end{proof}

\begin{lemma} \label{lma:schmidt}
Let $\psi$ be a state on a cone $C$ which is $f$-close far from the origin of $C$ to a pure factorized state $\psi_0$. Then its density matrix (in the GNS Hilbert space of this factorized state) has eigenvalues with $g(r)$-decay for some $g(r)=\Or$ that depends only on $f(r)$. Conversely, for any density matrix on a cone $C$ (in the GNS Hilbert space of a pure factorized state) whose eigenvalues have $g(r)$-decay there is a state on that cone which has the same eigenvalues and is $f$-close far from the origin of $C$ to this pure factorized state. Furthermore, one can choose $f(r)$ so that it depends only on $g(r)$. 
\end{lemma}

\begin{proof}
The proof word to word repeats the proof of Lemma 4.4 from \cite{KSY}.
\end{proof}

Let $I_{n,m} := \{(x,y) \in [n,m) \times [0,+\infty)\}$, $I'_{n,m} := \{(x,y) \in [n,m) \times (-\infty,0)\}$, and let $C^{-}_n : = I_{-\infty,n}$, $C^{+}_n : = I_{n,+\infty}$, $C'^{-}_n : = I'_{-\infty,n}$, $C'^{+}_n : = I'_{n,+\infty}$, $B^{\pm}_n = C^{\pm}_n \cup C'^{\pm}_n$. We further split cones $C^{\pm}_n$ into two cones $D^{\pm}_{n}$, $E^{\pm}_{n}$ (and similarly for $C'^{\pm}_n$) as shown on Fig. \ref{fig:fig1}.

Let $\psi = \alpha_{Q|_A}^{-1}(\psi_0)$, $\psi' = \alpha_{Q|_{\bar{A}}}^{-1}(\psi_0)$ which coincide with $\psi_0$ on $\bar{A}$ and $A$, respectively, and therefore the state $\Psi := \psi|_A \otimes \psi'|_{\bar{A}}$ is pure. By Lemma \ref{lma:decomp} the state $\Psi$ is LGA equivalent to $\psi_0$ by an element of $\SG_{\p A}$. Let $\alpha_F \in \SG_{\p A}$ be an LGP, such that $\alpha_F(\Psi) = \psi_0$ and let $f(r)$ be a function, such that all $F_j$ and all $Q_j$ are $f$-localized.

Since $(Q_A - Q|_A) \in \SC^0_{\p A}$, by Lemma \ref{lma:fafa} $\psi=\alpha_{Q_A}^{-1}(\psi_0)$ and $\alpha_{Q|_A}^{-1}(\psi_0)$ are LGA equivalent by an element of $\SG_{\p A}$. Therefore it is enough to show that $\psi$ is SRE defect state.

We say that $\phi^{\pm}_n$ is a truncation of $\psi$ if it is a pure state, such that $\phi^{\pm}_n|_{D^{\pm}_n} = \psi|_{D^{\pm}_n}$, $\phi^{\pm}_n|_{\bar{C}^{\pm}_n} = \psi_0|_{\bar{C}^{\pm}_n}$ and with $\phi^{\pm}_n|_{E^{\pm}_n}$ being $h$-close to $\psi_0$ on $E^{\pm}_n$ far from $(n,0)$ for some MDP function $h(r)=\Or$. Similarly, we define truncations $\phi'^{\pm}_n$ of $\psi'$ (see Fig. \ref{fig:fig1}).

\begin{lemma}
The truncations $\phi^{\pm}_n$ ($\phi'^{\pm}_n$) of $\psi$ ($\psi'$) exist with $h(r)$ depending only on $f(r)$.
\end{lemma}

\begin{proof}
Clearly, it is enough to show that for $\phi^{-}_0$ and $\psi$.

Let $D = D^{-}_0 \cup D'^{-}_0$. First, let us show that the state $\Psi$ has the split property\footnote{that was introduced for 1d systems in \cite{Matsui}} along the cut $\p D$, i.e. $\Psi$ is quasi-equivalent to $\Psi|_{D} \otimes \Psi|_{\bar{D}}$. Let $\tilde{\Psi} := (\alpha_{F|_D} \circ \alpha_{F|_{\bar{D}}}) (\Psi)$ which is by Lemma \ref{lma:decomp} related to $\psi_0$ by an almost local at $(0,0)$ unitary observable. By Cor. 2.6.11 from \cite{bratteli2012operator} both states $\tilde{\Psi}$ and $\tilde{\Psi}|_{D} \otimes \tilde{\Psi}|_{\bar{D}}$ are quasi-equivalent to $\psi_0$, and therefore are quasi-equivalent to each other. Since LGAs preserve quasi-equivalence, it implies quasi-equivalence of $\Psi$ and $\Psi|_{D} \otimes \Psi|_{\bar{D}}$. Thus, the GNS Hilbert space associated with $\Psi$ factorizes into the tensor product of Hilbert spaces for cones $D$ and $\bar{D}$, and $\Psi|_{D}$ and $\Psi|_{\bar{D}}$ can be described by density matrices in these Hilbert spaces.

Since $\Psi$ and $\alpha_{F|_{\bar{D}}}(\Psi)$ are pure and coincide on $D$, and since $\alpha_{F|_{\bar{D}}}(\Psi)|_{\bar{D}}$ is close to a factorized state $\psi_0$ on $D$ far from (0,0), Lemma \ref{lma:schmidt} implies that the density matrix of $\Psi|_D$ has $g(r)$-decay of Schmidt coefficients for some $g(r)$ that depends on $f(r)$ only. Since $\Psi = \psi|_{A} \otimes \psi'|_{\bar{A}}$, the state $\psi|_{D^{-}_0}$ is described by a density matrix with $g(r)$-decay of Schmidt coefficients, and therefore by Lemma \ref{lma:schmidt} it can be purified by a state on $D^{-}_0 \cup E^{-}_0$ which is $h(r)$-close on $E^{-}_0$ to $\psi_0$ far from (0,0) for some $h(r)$ that depends on $f(r)$ only, that guarantees the existence of $\phi^{-}_0$. 

\end{proof}

\begin{lemma} \label{lma:almfacgluing}
If a state is $g$-close to a pure factorized state $\psi_0$ both on a cone $C_1$ and on a cone $C_2$ with the same origin $j$, then it is $h$-close to $\psi_0$ on $C_1 \cup C_2$ for some MDP function $h(r)$ that depends only on $g(r)$.
\end{lemma}

\begin{proof}
The proof is the same as the proof of Lemma 4.3 from \cite{KSY} with a replacement of intervals $(n,+\infty)$ and $(-n,-\infty)$ by $C \cup \bar{\Gamma}_n(j)$ and $\bar{C} \cup \bar{\Gamma}_n(j)$, respectively.
\end{proof}

Let $\tilde{\psi}_n := \phi^{-}_n|_{B^{-}_n} \otimes \phi^{+}_n|_{B^{+}_n}$ and $\tilde{\psi}'_n := \phi'^{-}_n|_{B^{-}_n} \otimes \phi'^{+}_n|_{B^{+}_n}$.

\begin{lemma} \label{lma:A2}
States $\psi$ and $\tilde{\psi}_n$ (or $\psi'$ and $\tilde{\psi}'_n$) are related by an almost local at $(n,0)$ unitary $\CU \in \SA_{A}$ with localization depending on $f(r)$ only.
\end{lemma}
\begin{proof}

Let $\tilde{\Psi}_n := \tilde{\psi}_n|_{A} \otimes \psi'|_{\bar{A}} = \alpha_{Q|_{\bar{A}}}^{-1} (\tilde{\psi}_n)$ and let $E_n = E^-_n \cup E^+_n$. First, note that $\alpha_{F|_{\bar{E}_n}} (\tilde{\Psi}_n)$ and $\alpha_{F|_{\bar{E}_n}} (\Psi)$ are $g$-close to $\psi_0$ on cones $E_n$, $D^-_n \cup D'^-_n$, $E'^-_n \cup E'^+_n$, $D^+_n \cup D'^+_n$, and therefore by Lemma \ref{lma:almfacgluing} are $h$-close to $\psi_0$ on $\RR^2$ for some MDP function $h(r)=\Or$ that depends only on $f(r)$. By Lemma \ref{lma:almfac} they are related by a conjugation with an almost local unitary, and therefore the same is true for states $\tilde{\psi}_n$ and $\psi$.

The state $\alpha_{Q|_{A}} (\tilde{\psi}_n)$ can be produced from $\psi_0 = \alpha_{Q|_{A}} (\psi)$ by an almost local unitary $\tilde{\CU}$ and coincides with $\psi_0$ on $\bar{A}$. Therefore this unitary can be chosen to be in $\SA_{A}$. Thus, $\tilde{\psi}_n$ can be obtained from $\psi$ by a conjugation with $\CU = \alpha_{Q|_A}^{-1}(\tilde{\CU}) \in \SA_{A}$.
\end{proof}

\begin{lemma} \label{lma:1dLGA}
For any MDP function $f(r)=\Or$ there is $L$ such that any ordered composition $\overrightarrow{\prod_{n=-\infty}^{\infty}} \alpha_{\CB^{(n)}}$ of LGPs generated by  $f$-local at $(n L,0)$ observables $\CB^{(n)}$ for $n \in \ZZ$ is an LGP.
\end{lemma}

\begin{proof}
First, note that it is enough to show this for $\overrightarrow{\prod_{n=0}^{\infty}} \alpha_{\CB^{(n)}}$. Second, to prove the latter it is enough to show that with an  appropriate choice of $L$ for any $f$-local at 0 observable $\CA$ the observable $(\overleftarrow{\prod_{n=1}^{N}} \alpha_{\CB^{(n)}})(\CA)$ is almost local at 0 with localization depending on $f$ only (in particular, independent of $N$), and that as $N \to \infty$ it converges in the norm to some element of $\SA$.

Let $B_n(r):=I_{n-r,n+r} \cup I'_{n-r,n+r}$. Let $\CU^{(n)} := e^{i \CB^{(n)}}$ be a unitary that corresponds to $\alpha_{\CB^{(n)}}$. It can be represented as a product $(\CV^{(n)}_0 \CV^{(n)}_1 \CV^{(n)}_2...)$ of strictly local unitaries $\CV^{(n)}_k$ on $B_n((k+\frac12)L)$, so that $\|\CV^{(n)}_k-1\| \leq h((k+\frac12)L)$ for some MDP function $h(r)=\Or$ that depends on $f(r)$ only. This is achieved by letting $(\CV^{(n)}_0 \CV^{(n)}_1 ... \CV^{(n)}_k) = e^{i\CB^{(n)}|_{B_n((k+\frac12)L)}}$.

Since conjugation of a strictly local observable $\CA$ with a unitary $\CU$ strictly local in the localization set  of $\CA$ does not change the property $\|\CA-1\|<\eps$ and preserves the localization set, we can rearrange unitaries $\CV^{(n)}_k$ in the product 
\beq \label{eq:firstpr}
(\CV^{(1)}_0 \CV^{(1)}_1 \CV^{(1)}_2...) (\CV^{(2)}_0 \CV^{(2)}_1 \CV^{(2)}_2...) ... (\CV^{(N)}_0 \CV^{(N)}_1 \CV^{(N)}_2...)
\eeq
in the following order:
\beq \label{eq:secondpr}
(\tilde{\CV}^{(1)}_0) (\tilde{\CV}^{(2)}_0 \tilde{\CV}^{(1)}_1) (\tilde{\CV}^{(3)}_0 \tilde{\CV}^{(2)}_1 \tilde{\CV}^{(1)}_2) ... (\tilde{\CV}^{(n)}_0 \tilde{\CV}^{(n-1)}_1 ... \tilde{\CV}^{(1)}_{n-1}) ...,
\eeq 
where $\tilde{\CV}^{(n)}_k$ is obtained from   $\CV^{(n)}_k$ by conjugation with $\CV^{(m)}_l$ with $m,l$ satisfying $n+1\leq m \leq n+k$ and $0\leq l\leq n+k-m$. Importantly,  $\tilde{\CV}^{(n)}_k$ is strictly local on the same interval as $\CV^{(n)}_k$ and still satisfies $\|\tilde{\CV}^{(n)}_k - 1\| \leq h((k+\frac12)L)$. Therefore after this rearrangement the infinite product eq. (\ref{eq:secondpr}) still converges to the same unitary observable as eq. (\ref{eq:firstpr}).

Let $\tilde{\CU}^{(n)} = \tilde{\CV}^{(n)}_0...\tilde{\CV}^{(1)}_{n-1}$. We can represent $\CA = \sum_{p=0}^{\infty} \CA_p$ with $\sum_{p=0}^{n} \CA_p = \CA|_{B_{0}(n+1/2)}$. Let $\CA^{(0)}_p:=\CA_p$, $\CA^{(n)}_p:=\sum_{k=0}^{n-1} \tilde{\CU}^{(n)*} [\CA^{(k)}_p, \tilde{\CU}^{(n)}]$. Note that $\CA^{(n)}_p \in \SA|_{B_0(p+\frac12)}$ for $n \leq p$, and $\CA^{(n)}_p \in \SA|_{B_0(n+\frac12)}$ for $n>p$. Therefore we have
\begin{multline} \label{eq:Aestimate}
\|\CA^{(n)}_0\| =  \sum_{k=0}^{n-1} \|[\CA^{(k)}_0,\tilde{\CU}^{(n)}]\| \leq \sum_{k=0}^{n-1} \sum_{l>\frac{n-k-1}{2}}^{n-1} \|[\CA^{(k)}_0,\tilde{\CV}^{(n-l)}_l]\| \leq \\ \leq \sum_{k=0}^{n-1} \sum_{l > \frac{n-k-1}{2}}^{\infty}  2 \|\CA^{(k)}_0\| h((l+\frac12)L) \leq 2 \sum_{k=0}^{n-1} \|\CA^{(k)}_0\| g_{n-k}
\end{multline}
where $g_n := g(n L/2)$ for $g(n) := \sum_{l \geq n}^{\infty} h(l)$.

Any MDP function $g(r)=\Or$ can be upper-bounded by a reproducing MDP function $\tilde{g}=\Or$ \cite{hastings2010quasi} for lattice $\Lambda \subset \RR^d$, i.e. an $\Or$ MDP function satisfying
\beq
\sup_{j,k \in \Lambda} \sum_{l \in \Lambda} \frac{\tilde{g}(|j-l|)\tilde{g}(|l-k|)}{\tilde{g}(|j-k|)} < \infty.
\eeq
We can further upper-bound $\tilde{g}(r)$ by a reproducing $\Or$ MDP function $g'(r)=A \tilde{g}(r)^{\alpha}/r^{\nu} = \Or$ for some constants $A$, $0<\alpha<1$ and $\nu > d$. Since $1/r^{\nu}$ is also reproducing, we have

\beq
A^2 \sum_{k=1}^{n-1} \frac{\tilde{g}(k L/2)^{\alpha} \tilde{g}((n-k) L/2)^{\alpha}}{(kL/2)^{\nu} ((n-k)L/2)^{\nu}} < \frac{C}{L^{\nu}} A \frac{\tilde{g}(n L/2)^{\alpha}}{(nL/2)^{\nu}},
\eeq
and therefore for $g'_n := g'(nL/2) \geq g_n$ we have $\sum_{k=1}^{n-1} g'_k g'_{n-k} < (C/L^{\nu}) g'_n$
for some constant $C$. By changing $L$ one can make $(C/L^{\nu})<1/2$. Therefore one can choose $L$ so that for $a_n = 2 n g'_n$ we have 
\begin{multline} 
2(g_n \cdot 1 + g_{n-1} a_{1}+ g_{n-2} a_{2} + ... + g_{1} a_{n-1}) \leq \\ \leq 2(g'_n \cdot 1 + 2 g'_{n-1} g'_{1} + 4 g'_{n-2} g'_{2} + ... + 2(n-1) g'_{1} g'_{n-1}) \leq 2 n g'_n = a_n.
\end{multline}
Together with eq. (\ref{eq:Aestimate}) this implies that $\|\CA^{(n)}_0\|/\|\CA_0\|$ can be upper-bounded by $a_n = \On$, and the sequence $\sum_{k=0}^{n}\CA^{(k)}_0$ converges in the norm to some element, which is almost local at 0. By construction the localization depends on $f$ only.

In the same way one can estimate the norms of $\CA^{(n+p)}_p$ for $n,p>0$ and bound $\|\CA^{(n+p)}_p\|/\|\CA_p\|$ by a sequence $a_n = \On$. Together with $\|\CA_p\| = \| \sum_{q=0}^{p}\CA^{(q)}_p\|$ that ensures convergence of $\sum_{p=0}^{\infty} \sum_{n=0}^{\infty} \CA^{(n)}_p$ to some almost local at (0,0) observable with localization depending on $f$ only.

\end{proof}

\begin{lemma} \label{lma:A3}
The states $\phi^{\pm}_n$ (or $\phi'^{\pm}_n$) are SRE defect states and are produced by an LGP $\alpha_{G}$ with $G_j \in \SA|_{C^{\pm}_n}$ (or $G_j \in \SA|_{C'^{\pm}_n}$).
\end{lemma}

\begin{proof}
Let us define the truncations $\chi^{\pm}_{n,m}$ of $\phi^{+}_n$ in the same way as the truncations $\phi^{\pm}_m$ are defined for $\psi$. Note that for $n<m$ the state $\chi^{+}_{n,m}$ coincides with $\phi^{+}_m$. 

Lemma \ref{lma:A2} with a replacement of $\psi$ and $\psi'$ by $\phi^{+}_n$ and $\phi'^{+}_n$, respectively, implies, that for any $L>0$ the states $\phi^{+}_n$ and $\chi^{-}_{n,n+L}|_{B^{-}_{n+L}} \otimes \chi^{+}_{n,n+L}|_{B^{+}_{n+L}}$ are related by an almost local at $(n+L,0)$ unitary $\CU^{(1)} \in \SA_{A}$, with localization function depending only on $f$. Therefore by a consecutive application of conjugations with such unitaries $\CU^{(k)}$ at points $(n+kL,0)$ one can produce a state, which is a tensor product of $\psi_0|_{\bar{A}}$ and states on $I_{n+k L,n+(k+1)L}$, each of which by Lemma \ref{lma:almfacgluing} is $h$-close to $\psi_0$ far from $n+(k+\frac12)L$ for some MDP function $h(r)=\Or$ that depends only on $f$ and $L$. Since localization function of all unitaries $\CU^{(n)}$ is the same and depends only on $f$, by Lemma \ref{lma:1dLGA} one can find $L$, such that an ordered product of conjugations with such unitaries is an LGA. The resulting state is LGA equivalent to $\psi_0$ by Lemma \ref{lma:almfac}.

Similarly, one can show that $\phi^{-}_n$ and $\phi'^{\pm}_n$ are SRE defect states.
\end{proof}

Lemma \ref{lma:A3} implies that $\tilde{\psi}_n$ is an SRE defect state. Since by Lemma \ref{lma:A2} it is related to $\psi$ by an almost local unitary at $(n,0)$, we conclude that $\psi$ is an SRE defect state.

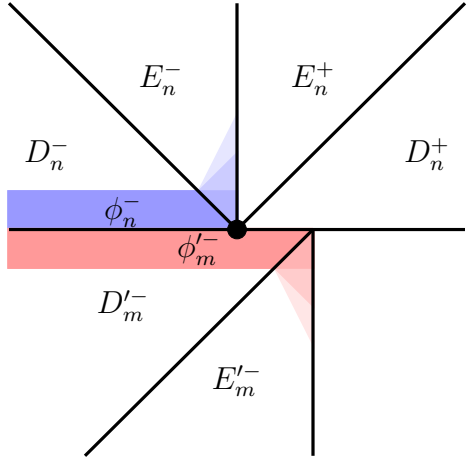
\begin{figure}
\centering
\begin{tikzpicture}[scale=.5]

\draw [color=blue!10, fill=blue!10, very thick] (0,0) to (-6,0) -- (-6,1) -- (-1,1) -- (0,3) -- (0,0);
\draw [color=blue!20, fill=blue!20, very thick] (0,0) to (-6,0) -- (-6,1) -- (-1,1) -- (0,2) -- (0,0);
\draw [color=blue!30, fill=blue!30, very thick] (0,0) to (-6,0) -- (-6,1) -- (-1,1) -- (0,1) -- (0,0);
\draw [color=blue!40, fill=blue!40, very thick] (0,0) to (-6,0) -- (-6,1) -- (-1,1) -- (0,0);

\draw [color=red!10, fill=red!10, very thick] (2,0) to (-6,0) -- (-6,-1) -- (-1+2,-1) -- (0+2,-3) -- (0+2,0);
\draw [color=red!20, fill=red!20, very thick] (0+2,0) to (-6,0) -- (-6,-1) -- (-1+2,-1) -- (0+2,-2) -- (0+2,0);
\draw [color=red!30, fill=red!30, very thick] (0+2,0) to (-6,0) -- (-6,-1) -- (-1+2,-1) -- (0+2,-1) -- (0+2,0);
\draw [color=red!40, fill=red!40, very thick] (0+2,0) to (-6,0) -- (-6,-1) -- (-1+2,-1) -- (0+2,0);

% \draw[black, very thick] (0,0) -- (0,-3);
\draw[black, very thick] (0,0) -- (0,6);
\draw[black, very thick] (0,0) -- (-6,6);
\draw[black, very thick] (0,0) -- (6,6);
\draw[black, very thick] (0,0) -- (6,0);
\draw[black, very thick] (0,0) -- (-6,0);
\draw[black, very thick] (2,0) -- (2,-6);
\draw[black, very thick] (2,0) -- (2-6,-6);

\filldraw[black] (0,0) circle (7pt) node[anchor=west]{};

\node  at (-5,2) {$D^{-}_n$};

\node  at (-2,4) {$E^{-}_n$}; 

\node  at (-3,.5) {$\phi^{-}_n$};

\node  at (5,2) {$D^{+}_n$};

\node  at (2,4) {$E^{+}_n$}; 

% \node  at (3,.5) {$\phi^{+}_n$};

\node  at (-5+2,-2) {$D'^{-}_{m}$};

\node  at (-2+2,-4) {$E'^{-}_{m}$}; 

\node  at (-3+2,-.5) {$\phi'^{-}_m$};

% \node  at (-.5,-.5) {$p$}; 
% \node  at (2,1.5) {$A$};
% \node  at (0,4.5) {$(\p A)_-$};
% \node  at (1.3*3.4641,-1.3*2) {$(\p A)_+$};

\end{tikzpicture}
\caption{A truncation $\phi^{-}_n$ of the state $\psi$ to $C^{-}_n$, that coincides with $\psi$ inside the cone $D^-_n$ and is almost factorized state inside the cone $E^{-}_n$. The cones $D^{-}_n$ and $E^{-}_n$ split the cone $C^{-}_n$ into two.
}
\label{fig:fig1}
\end{figure}

\printbibliography

@article{kitaev2006anyons,
  title={Anyons in an exactly solved model and beyond},
  author={Kitaev, Alexei},
  journal={Annals of Physics},
  volume={321},
  number={1},
  pages={2--111},
  year={2006},
  publisher={Elsevier}
}

@article{kapustin2019thermal,
  title={Thermal Hall conductance and a relative topological invariant of gapped two-dimensional systems},
  author={Kapustin, Anton and Spodyneiko, Lev},
  journal={arXiv preprint arXiv:1905.06488},
  year={2019}
}

@article{Matsui,
  title={Boundedness of entanglement entropy and split property of quantum spin chains},
  author={Matsui, Taku},
  journal={Reviews in Mathematical Physics},
  volume={25},
  number={09},
  pages={1350017},
  year={2013},
  publisher={World Scientific}
}

@misc{bourne2020classification,
      title={The classification of symmetry protected topological phases of one-dimensional fermion systems}, 
      author={Chris Bourne and Yoshiko Ogata},
      year={2020},
      eprint={2006.15232},
      archivePrefix={arXiv},
      primaryClass={math-ph}
}

@article{hastings2005quasiadiabatic,
  title={Quasiadiabatic continuation of quantum states: The stability of topological ground-state degeneracy and emergent gauge invariance},
  author={Hastings, Matthew B and Wen, Xiao-Gang},
  journal={Physical review b},
  volume={72},
  number={4},
  pages={045141},
  year={2005},
  publisher={APS}
}

@book{bratteli2012operator,
    AUTHOR = {Bratteli, Ola and Robinson, Derek W.},
     TITLE = {Operator algebras and quantum statistical mechanics. 1. $C^\ast$- and $W^\ast$-algebras, symmetry groups,
              decomposition of states.},
    SERIES = {Texts and Monographs in Physics},
   EDITION = {2nd ed.},
 PUBLISHER = {Springer-Verlag, New York},
      YEAR = {1987},
}

@book {bratteli2012operator2,
    AUTHOR = {Bratteli, Ola and Robinson, Derek W.},
     TITLE = {Operator algebras and quantum statistical mechanics. 2. Equilibrium states. Models in quantum statistical mechanics},
    SERIES = {Texts and Monographs in Physics},
   EDITION = {2nd ed.},
 PUBLISHER = {Springer-Verlag, Berlin},
      YEAR = {1997},
}

@article{hastings2010quasi,
  title={Quasi-adiabatic continuation for disordered systems: Applications to correlations, Lieb-Schultz-Mattis, and Hall conductance},
  author={Hastings, Matthew B},
  journal={arXiv:1001.5280},
  year={2010}
}

@article{chen2011classification,
  title={Classification of gapped symmetric phases in one-dimensional spin systems},
  author={Chen, Xie and Gu, Zheng-Cheng and Wen, Xiao-Gang},
  journal={Physical review b},
  volume={83},
  number={3},
  pages={035107},
  year={2011},
  publisher={APS}
}

@book {QImeetsQM,
    AUTHOR = {Zeng, Bei and Chen, Xie and Zhou, Duan-Lu and Wen, Xiao-Gang},
     TITLE = {Quantum information meets quantum matter: From quantum entanglement to topological phases of many-body
              systems.},
    SERIES = {Quantum Science and Technology},
 PUBLISHER = {Springer, New York},
      YEAR = {2019},
}

@misc{ogata2019classification,
      title={A classification of pure states on quantum spin chains satisfying the split property with on-site finite group symmetries}, 
      author={Yoshiko Ogata},
      year={2019},
      eprint={1908.08621},
      archivePrefix={arXiv},
      primaryClass={math.OA}
}

@article{osborne2007simulating,
  title={Simulating adiabatic evolution of gapped spin systems},
  author={Osborne, Tobias J},
  journal={Physical review a},
  volume={75},
  number={3},
  pages={032321},
  year={2007},
  publisher={APS}
}

@article{kapustin2020hall,
  title={Hall conductance and the statistics of flux insertions in gapped interacting lattice systems},
  author={Kapustin, Anton and Sopenko, Nikita},
  journal={Journal of Mathematical Physics},
  volume={61},
  number={10},
  pages={101901},
  year={2020},
  publisher={AIP Publishing LLC}
}

@article{chen2013symmetry,
  title={Symmetry protected topological orders and the group cohomology of their symmetry group},
  author={Chen, Xie and Gu, Zheng-Cheng and Liu, Zheng-Xin and Wen, Xiao-Gang},
  journal={Physical Review B},
  volume={87},
  number={15},
  pages={155114},
  year={2013},
  publisher={APS}
}

@article{fidkowski2020exactly,
  title={Exactly solvable model for a 4+ 1 D beyond-cohomology symmetry-protected topological phase},
  author={Fidkowski, Lukasz and Haah, Jeongwan and Hastings, Matthew B},
  journal={Physical Review B},
  volume={101},
  number={15},
  pages={155124},
  year={2020},
  publisher={APS}
}

@article{KSY,
  title={A classification of phases of bosonic quantum lattice systems in one dimension},
  author={Kapustin, Anton and Sopenko, Nikita and Yang, Bowen},
  journal={arXiv preprint arXiv:2012.15491},
  year={2020}
}

@article{else2014classifying,
  title={Classifying symmetry-protected topological phases through the anomalous action of the symmetry on the edge},
  author={Else, Dominic V and Nayak, Chetan},
  journal={Physical Review B},
  volume={90},
  number={23},
  pages={235137},
  year={2014},
  publisher={APS}
}

@article{kapustin2014symmetry,
  title={Symmetry protected topological phases, anomalies, and cobordisms: beyond group cohomology},
  author={Kapustin, Anton},
  journal={arXiv preprint arXiv:1403.1467},
  year={2014}
}

@article{freed2016reflection,
  title={Reflection positivity and invertible topological phases},
  author={Freed, Daniel S and Hopkins, Michael J},
  journal={arXiv preprint arXiv:1604.06527},
  year={2016}
}

@article{gu2009tensor,
  title={Tensor-entanglement-filtering renormalization approach and symmetry-protected topological order},
  author={Gu, Zheng-Cheng and Wen, Xiao-Gang},
  journal={Physical Review B},
  volume={80},
  number={15},
  pages={155131},
  year={2009},
  publisher={APS}
}

@article{ogata2020mathbb,
  title={A $\mathbb{Z}_2$-Index of Symmetry Protected Topological Phases with Time Reversal Symmetry for Quantum Spin Chains},
  author={Ogata, Yoshiko},
  journal={Communications in Mathematical Physics},
  volume={374},
  number={2},
  pages={705--734},
  year={2020},
  publisher={Springer}
}

@unpublished{Kitaevlecture,
author  = {Alexei Kitaev},
title    = {On the classificaton of Short-Range Entangled states},
note = {Talk at Simons Center for Geometry and Physics, June 2013},
URL={http://scgp.stonybrook.edu/archives/7874}
}

@article{ogata2021h,
  title={A $ H^{3}(G,{\mathbb T})$-valued index of symmetry protected topological phases with on-site finite group symmetry for two-dimensional quantum spin systems},
  author={Ogata, Yoshiko},
  journal={arXiv preprint arXiv:2101.00426},
  year={2021}
}

\end{document}